\documentclass{article}

\usepackage[utf8]{inputenc}
\usepackage{float}

\usepackage{amsthm}
\usepackage[all]{xy}
\usepackage{amsmath}
\usepackage{amssymb}
\usepackage{amsfonts}
\usepackage{relsize}
\usepackage{verbatim}

\usepackage{enumerate,tikz}
\usetikzlibrary{arrows}
\usetikzlibrary{matrix}

\newtheorem{lemma}{Lemma}

\newcommand{\psf}{\mathcal{P}}

 \newtheorem{theorem}{Theorem}[section]
 \newtheorem{proposition}[theorem]{Proposition}
 \newtheorem{corollary}[theorem]{Corollary}

 \theoremstyle{definition}
 \newtheorem{definition}[theorem]{Definition}

 \theoremstyle{remark}

\title{Generalized Vietoris Bisimulations}
\author{Sebastian Enqvist \and Sumit Sourabh}

\begin{document}
\maketitle
\begin{abstract}
We introduce and study bisimulations for coalgebras on Stone spaces \cite{KKV05}. Our notion of bisimulation is sound and complete for behavioural equivalence, and generalizes Vietoris bisimulations \cite{BeFoVe10}. The main result of our paper is that bisimulation for a \textbf{Stone} coalgebra is the topological closure of bisimulation for the underlying \textbf{Set} coalgebra.  
\end{abstract}
\section{Introduction}

The notion of \emph{bisimulation} plays an important role in several areas of computer science and mathematics: concurrency theory \cite{Milner1982,Milner91,Park81} , modal logic \cite{vanBenthem83,BRV01}, formal verification \cite{Cleaveland96} and set theory \cite{Aczel88,forti1983set}. For a recent survey on bisimulation, we refer to \cite{Sangiorgi2009}. More generally, \textit{coalgebras} provide a uniform framework for  studying the behaviour of state-based systems \cite{Jacobs97atutorial,Rutten2003}, and bisimulations between coalgebras play an important role in this context as proof methods for \textit{behavioural equivalence}.
In this paper, we introduce and study bisimulations for \emph{Stone coalgebras} \cite{KKV05}, that is, coalgebras for which the state space is equipped with a Stone topology.

Stone coalgebras have previously been studied mainly in the context of coalgebraic modal logic, where they generalize descriptive general frames. From the perspective of algebraic duality, descriptive general frames are more natural than standard Kripke frames, and it is well known that every normal modal logic is complete with respect to a class of descriptive general frames \cite{BRV01}, unlike Kripke frames.  Descriptive general frames are isomorphic to coalgebras for the Vietoris functor on $\mathbf{Stone}$ \cite{KKV05}, which makes Stone coalgebras important in coalgebraic modal logic. Moreover, since the category of coalgebras for Vietoris functor on $\mathbf{Stone}$ is dual to the category of modal algebras, they provide a natural semantics for \emph{finitary modal logics} \cite{KKV05}.

In \cite{KKP05},  Kupke, Kurz and Pattinson give a more general construction for lifting an endofunctor on the category $\mathbf{Set}$ to an endofunctor on the category $\mathbf{Stone}$, by making use of  a set of \textit{predicate liftings} for the functor. If the endofunctor on $\mathbf{Set}$ is the (covariant) power set functor, then the lifted endofunctor on $\mathbf{Stone}$, using their construction, is the Vietoris functor \cite{johnstone1986}. Here, we shall consider a slight variation of this construction, using the Moss-style $\nabla$-modality \cite{Moss99} for $T$ rather than predicate liftings.  In this respect our approach is also similar to the construction in \cite{VVV13}, which  introduces an endofunctor $V_T$ on the category of frames, parametrized by an endofunctor $T$ on the category $\mathbf{Set}$ that satisfies certain constraints. 

A notion of bisimulation for descriptive general frames, called \emph{Vietoris bisimulation}, was defined in \cite{BeFoVe10}. In this paper, we generalize this work by introducing a notion of bisimulation for coalgebras on Stone spaces. To the best of our knowledge, bisimulations for colagebras on topological spaces have not been studied before, and this paper aims to fill this gap. Given an endofunctor $T:\mathbf{Set}\to\mathbf{Set}$, we introduce a notion of bisimulation, called \emph{neighbourhood bisimulation}, on Stone coalgebras \cite{KKV05} for the lifted functor $\widehat{T}:\mathbf{Stone}\to\mathbf{Stone}$. (The reason for the name ``neighbourhood bisimulation'' is that the concept bears some resemblance to bisimulations studied in the context of neighbourhood semantics for modal logic in \cite{HaKuPa09}.) In the case where $T$ is the powerset functor, closed neighbourhood bisimulations are shown to coincide with Vietoris bisimulations. We also show that neighbourhood bisimulations  are always sound and complete with respect to behavioural equivalence. 


The main result of our paper is that the topological closure of a neighbourhood bisimulation between two $\widehat{T}$-coalgebras is always a  neighbourhood bisimulation too. Since closed neighbourhood bisimulations generalize Vietoris bisimulations, the main result in \cite{BeFoVe10} stating that the closure of a Kripke bisimulation is a Vietoris bisimulation follows as an immediate corollary to our result. Moreover, our proof is simpler than the one presented in \cite{BeFoVe10}.  It also follows that the topologically closed neighbourhood bisimulations between any pair of $\widehat{T}$-coalgebras form a complete lattice ordered by inclusion, again generalizing a result for Vietoris bisimulations proved in \cite{BeFoVe10}. 

\section{Preliminaries}
This section introduces the basic concepts from coalgebra and coalgebraic modal logic that we will need here. Readers that are already familiar with these fields can skip it. Familiarity with basic category theory (functors, natural transformations, adjunctions etc.) will be assumed throughout the paper. 
\subsection{Coalgebras}
 Coalgebras for a functor are used in theoretical computer science as models for various sorts of process-like systems: labelled transition systems \cite{Plotkin04}, automata \cite{Venema2006}, probabilistic transition systems \cite{deVink1999} etc.  Given any endofunctor on a category, the corresponding category of coalgebras for this functor is constructed as follows:
\begin{definition}
Let $\mathbf{C}$ be a category and let $T :\mathbf{C} \rightarrow \mathbf{C}$ be an endofunctor. A \emph{$T$-coalgebra} is a pair $(X, \sigma)$, where $\sigma : X\rightarrow T X$ is a morphism in $\mathbf{C}$. A morphism between two coalgebras $(X, \sigma)$ and $(X',\sigma')$ is a morphism $f$ in $\mathbf{C}$ such that $\sigma\circ f= Tf\circ\sigma$, i.e. the following diagram commutes:
\begin{center}
\begin{tikzpicture}
  \matrix (m) [matrix of math nodes, row sep=3em,column sep=4em,minimum width=2em]
  {
     X & X' \\
     TX &TX' \\};
  \path[-stealth]
    (m-1-1) edge node [left] {$\sigma$} (m-2-1)
            edge node [above] {$f$} (m-1-2)
    (m-2-1.east|-m-2-2) edge node [below] {$Tf$} (m-2-2)
    (m-1-2) edge node [right] {$\sigma'$} (m-2-2);
\end{tikzpicture}
\end{center}

\end{definition}

If $\mathbf{C}$ is the category $\mathbf{Set}$ of sets and mappings on sets, and $T = \psf$, the contravariant powerset functor (with actions on morphisms defined by  $\psf h (Z) = h[Z]$), then the corresponding coalgebras are \textit{Kripke frames} and coalgebra morphisms are \textit{$p$-morphisms}, a fundamental concept in the model theory of modal logic. This has lead to the generalization from standard modal logic to \textit{coalgebraic modal logic} (see, e.g., \cite{KupkeP11,CirsteaEA08}) , where modal logics are viewed as specification languages for coalgebras, and each functor $T : \mathbf{Set} \rightarrow \mathbf{Set}$ comes with an associated modal logic. A common approach to coalgebraic modal logic uses \textit{predicate liftings} for a functor \cite{Schroder:2008}, while the original approach introduced by Lawrence moss \cite{Moss99} uses relation lifting  (see \cite{Le08} for a comparison). If $T$ preserves weak pullbacks, then the canonical choice of relation lifting is given by the \textit{Barr extension} of a functor, which assigns to each pair of sets $X$ and $X'$ and each binary relation $R\subseteq X\times X'$ the lifted relation
\[\overline{T}(R) := \{((T \pi)(\rho),(T \pi')(\rho)) :\rho\in T R\},
\] 
where $\pi: R \rightarrow X$ and $\pi': R \rightarrow X'$ are the projection maps. More generally, one can consider an arbitrary \textit{lax extension} $L$ for the functor $T$:
 \begin{definition}
A relation lifting $L$ for a set endofunctor $T$ is a \emph{lax extension} of $T$ if it satisfies the following conditions for all relations $R,R'\subseteq X\times Z$ and $S\subseteq Z\times Y$, and all functions $f:X\to Z$:
 \begin{description}
\item[L1:] $R'\subseteq R$ implies $LR'\subseteq LR$
\item[L2:] $LR ; LS\subseteq L(R;S)$
\item[L3:] $Tf\subseteq Lf$
\end{description}
A lax extension is said to be \emph{symmetric} if, for any relation $R$, we have $L(R^\dagger) = (LR)^\dagger$ (where the operator $(-)^\dagger$ takes a binary relation to its converse). 
\end{definition}
Any symmetric lax extension then provides a generalized ``nabla modality'' as follows:
 \begin{definition}
   Let $T$ be a covariant set functor. A distributive law of $T$ over a (co- or contravariant) set functor $M$ is a natural transformation $\nabla:TM\to MT$.
  \end{definition}
Consider the contravariant powerset functor $Q : \mathbf{Set} \rightarrow \mathbf{Set}$, which acts on objects as the covariant powerset functor, and has its action on a map $h : X \rightarrow Y$ defined by $Qh : Z \mapsto h^{-1}[Z]$, where $Z \in QY$. Any symmetric lax extension for $T$ gives rise to a distributive law $\nabla : TQ \rightarrow QT$ by the assignment
$$\nabla_X : \varphi \mapsto  \{\alpha \in TX \mid \alpha (L\in_X) \varphi \}$$
where $\varphi \in TQX$ and $\in_X \subseteq X \times QX$ is the membership relation. This distributive law can be interpreted as the semantics of a modal operator: if the subsets of $X$ are thought of as propositions over $X$, then a member $\varphi$ of $TQX$ can be seen as a modal formula built up from such propositions, and $\nabla_X \varphi \subseteq TX$ is a proposition over $TX$ which gives the interpretation of the formula. With this interpretation, it makes more sense to consider the \textit{finitary version} of $T$, denoted $T^\omega$. This functor sends a set $X$ to the set
$$\coprod \{T Y \mid Y \text{ a finite subset of } X\}$$
and the action of $T^\omega$ on morphisms should be fairly obvious. Obviously, we can construct the natural transformation $\nabla : T^\omega Q \rightarrow QT$ using the lax extension $L$ in the same manner  as before. This gives the semantics of \textit{finitary} modal formulas, built up from  finitely many propositions over a set $X$.

From now on we assume that we are given a functor $T : \mathbf{Set} \rightarrow \mathbf{Set}$ and an appropriate lax extension $L$ for it. In the case where $T$ preserves weak pullbacks, as in the case of the covariant powerset functor, the canonical choice would be the Barr extension. For a detailed overview on relation lifting, see e.g. \cite{KKV12}. For more on lax extensions and their role in coalgebraic modal logic, see e.g. \cite{MartiV12}. 

\subsection{Bisimulation and behavioural equivalence}

A basic concept in coalgebra theory is that of \textit{behavioural equivalence}:

\begin{definition}
Let $\mathbf{C}$ be any category equipped with a forgetful functor $U : \mathbf{C} \rightarrow \mathbf{Set}$, let $(X,\sigma)$ and $(X',\sigma')$ be coalgebras for a functor $T : \mathbf{C} \rightarrow \mathbf{C}$ and let $u \in UX$ and $u' \in UX'$. Then we say that $u$ and $u'$ are \emph{behaviourally equivalent} if there exists a $T$-coalgebra $(Y,\tau)$ and a pair of coalgebra maps $h : (X,\sigma) \rightarrow (Y,\tau)$ and $h' : (X',\sigma') \rightarrow (Y,\tau)$ such that $Uh(u) = Uh' (u')$.  
\end{definition}
The structure $(X,\sigma,u)$ in this definition is called a \textit{pointed} coalgebra, and we shall refer to $u$ as a ``state'' of the coalgebra.
When coalgebras represent some type of process or computation, such as labelled transition systems in concurrency theory, behaviourally equivalent states are taken to represent ``essentially the same process''. In non-wellfounded set theory, coalgebras for the powerset functor represent systems of equations, and the states are thought of as \textit{variables}.  Two variables from a pair of systems of equations are behaviourally equivalent iff they define the same set.  

The usual proof method for showing that two pointed coalgebras are behaviourally equivalent is by \textit{coinduction}: two states are behaviourally equivalent if we can establish a \textit{bisimulation} between them.  The relation lifting $L$ can be used to define a notion of bisimilarity for $T$-coalgebras:
 \begin{definition} An \emph{$L$-bisimulation} between $(X,\alpha)$ and $(Y,\beta)$ is a relation $R\subseteq X\times Y$ such that $(\alpha(x),\beta(y))\in LR$ for all $(x,y)\in R$. A state $x$ of $(X,\alpha)$ is \emph{$L$-bisimilar} to a state $y$ of $(Y,\beta)$ if there is an $R\subseteq X\times Y$ that is an $L$-bisimulation between $(X,\alpha)$ and $(Y,\beta)$ with $(x,y)\in R$.
 \end{definition}

If $L$-bisimilarity coincides with behavioural equivalence in every pair of coalgebras, then we say that $L$-bisimilarity is \textit{sound and complete} for behavioural equivalence. The class of functors that have a lax extension $L$ such that $L$-bisimilarity is sound and complete for behavioural equivalence have been characterized in \cite{MartiV12}.

Another useful method to establish a behavioural equivalence is by \textit{induction along the terminal sequence} of a functor. Let $T$ be any endofunctor on a category $\mathbf{C}$. If $\mathbf{C}$ has all small filtered colimits (and therefore an initial object), then the \textit{initial sequence} for $T$ can be constructed, and consists of an object $\mathcal{A}_\xi$ for each ordinal $\xi$, together with a unique morphism $h_\xi^\zeta : \mathcal{A}_\xi \rightarrow \mathcal{A}_\zeta$ for all ordinals $\xi \leq \zeta$. It has the following properties:
\begin{itemize}
\item $\mathcal{A}_0$ is the initial object in $\mathbf{C}$
\item More generally, for any limit ordinal $\xi$, $\mathcal{A}_\xi$ is the colimit of the (obviously filtered) diagram consisting of all the objects of the initial sequence below $\xi$, together with all the  morphisms of the initial sequence with domain and codomain below $\xi$. Given $\zeta \geq \xi$, the map $h^\zeta_\xi : \mathcal{A}_\xi \rightarrow \mathcal{A}_\zeta$ is the unique connecting map from $\mathcal{A}_\xi$ to $\mathcal{A}_\zeta$ as the vertex of the cocone consisting of all morphisms $h^\zeta_\rho$ for $\rho < \xi$.
\item For any pair of ordinals $\xi \leq \zeta$, we have $\mathcal{A}_{\xi + 1} = T(\mathcal{A}_\xi)$ and $h_{\xi + 1}^{\zeta + 1} = T(h_\xi^\zeta)$.
\end{itemize}
If, for some ordinal $\xi$, the map $h^{\xi+1}_\xi : \mathcal{A}_\xi \rightarrow T \mathcal{A}_\xi$ is an isomorphism, then we say that the initial sequence \textit{stabilizes} at the ordinal $\xi$. (The inverse of $h^{\xi + 1}_\xi$ will then provide the \textit{initial algebra} for the functor $T$.) It is easy to see that if $T$ is finitary ($T$ preserves filtered colimits) then the initial sequence for $T$ stabilizes at $\omega$.

Dually, if $\mathbf{C}$ has all small co-filtered limits (and hence a terminal object), the \textit{terminal sequence} for $T$ can be constructed and consists of objects $\mathcal{Z}_\xi$ for all ordinals $\xi$, together with a unique morphism $g^{\zeta}_\xi : \mathcal{Z}_{\zeta } \rightarrow \mathcal{Z}_{\xi}$ for each pair of ordinals $\xi \leq \zeta$ (note the reversed order). It has the following properties:
\begin{itemize}
\item $\mathcal{Z}_0$ is the terminal object in $\mathbf{C}$
\item More generally, for any limit ordinal $\xi$, $\mathcal{Z}_\xi$ is the limit of the (co-filtered) diagram consisting of all the objects of the terminal sequence below $\xi$, together with all the  morphisms of the terminal sequence with domain and codomain below $\xi$. Given $\zeta \geq \xi$, the map $h^\zeta_\xi : \mathcal{Z}_\zeta \rightarrow \mathcal{Z}_\xi$ is the unique connecting map from $\mathcal{Z}_\zeta$ to $\mathcal{Z}_\xi$, where $\mathcal{Z}_\zeta$ is considered as the vertex of the cone consisting of all morphisms $h^\zeta_\rho$ for $\rho < \xi$.
\item For any pair of ordinals $\xi \leq \zeta$, we have $\mathcal{A}_{\xi + 1} = T(\mathcal{A}_\xi)$ and $h_{\xi + 1}^{\zeta + 1} = T(h_\xi^\zeta)$.
\end{itemize}
If for some ordinal $\xi$, the morphism $h_{\xi}^{\xi + 1}$ is an isomorphism with inverse $(h_{\xi}^{\xi + 1})^{-1} : \mathcal{Z}_\xi \rightarrow T\mathcal{Z}_\xi$, then the pair $(\mathcal{Z}_\xi,(h_{\xi}^{\xi + 1}),(h_{\xi}^{\xi + 1})^{-1})$ is called the \textit{final coalgebra} for the functor $T$, and has the universal property that, for every $T$-coalgebra $(X,\alpha)$ there is a unique coalgebra morphism $f : X \rightarrow \mathcal{Z}_\xi$. Moreover, this unique coalgebra morphism can be constructed as follows:

Fix any $T$-coalgebra $(X,\alpha)$. Then, for every ordinal $\xi$, there is a map $beh_\xi^\alpha : X \rightarrow \mathcal{Z}_\xi$, called the \textit{behaviour map} at the ordinal $\xi$, with the property that for $\xi \leq \zeta$ we have
$$h_{\xi}^\zeta \circ beh^\alpha_\zeta = beh^\alpha_\xi$$
These maps are defined inductively as follows:
\begin{itemize}
\item For $\xi = 0$, there is only one choice for the map $beh_0^\alpha$ since $\mathcal{Z}_0$ is a terminal object
\item For $\xi$ a limit ordinal, the maps $beh_\rho^\alpha$ for $\rho < \xi$ will form a cone for the diagram consisting of all objects and morphisms in the terminal sequence below $\xi$, with vertex $X$. Hence, we can take $beh_\xi^\alpha$ to be the unique connecting map from $X$ to the limit $\mathcal{Z}_\xi$.
\item Given that we have constructed the map $beh^\alpha_\xi$, the map $beh^\alpha_{\xi + 1}$ is defined to be $T(beh^\alpha_\xi) \circ \alpha$.
\end{itemize}
If the terminal sequence stabilizes to yield a final $T$-coalgebra at $\xi$, then $beh_\xi^\alpha$ is the unique coalgebra morphism from $(X,\alpha)$ to the final coalgebra. If there is a forgetful functor $U : \mathbf{C} \rightarrow \mathbf{Set}$, this justifies the method of \textit{terminal sequence induction} to prove that two states $u,v$ from a pair of coalgebras $(X,\alpha)$ and $(Y,\beta)$ are behaviourally equivalent: it suffices to prove, by transfinite induction, that $beh_\xi^\alpha(u) = beh_\xi^\beta(v)$ for every ordinal $\xi$. 

\section{From Set Functors to Endofunctors on $\mathbf{Stone}$}
The category $\mathbf{Stone}$ is the category of Stone spaces (i.e. compact and totally disconnected spaces) with continuous maps as morphisms, and the category $\mathbf{BA}$ has Boolean algebras as objects and Boolean algebra homomorphisms as arrows. We fix notation for the forgetful functor $U : \mathbf{Stone} \rightarrow \mathbf{Set}$ and the forgetful functor $V : \mathbf{BA} \rightarrow \mathbf{Set}$. We consider the contravariant powerset functor as a contravariant functor $Q : \mathbf{Set} \rightarrow \mathbf{BA}$, and we also have a  contravariant functor $S : \mathbf{BA} \rightarrow \mathbf{Stone}$ sending each Boolean algebra to the Stone space of its ultrafilters (with the topology generated by the clopen basis of sets $\{F \in SA \mid a \in F\}$ for some $a \in A$). Finally, we have a  contravariant functor $P : \mathbf{Stone} \rightarrow \mathbf{BA}$ sending a Stone space to the Boolean algebra of its clopen subsets.

By Stone duality, the functors $S$ and $P$ constitute an equivalence of categories between $\mathbf{Stone}$ and $\mathbf{BA}^{op}$.  
%
%
%

It is well known that the covariant powerset functor on $\mathbf{Set}$ has a similar counterpart in the \textit{Vietoris functor} $\mathcal{V} : \mathbf{Stone} \rightarrow \mathbf{Stone}$. Its action on a Stone space $\mathbb{X}$ is to let $\mathcal{V} \mathbb{X}$ consist of all the closed sets of $\mathbb{X}$, with the topology generated by all sets of the form
$$(1) \quad\quad \Box Z = \{S \mid S \subseteq Z\}$$
for $Z$ clopen in $\mathbb{X}$, and all sets of the form
$$(2) \quad \quad \Diamond Z = \{S \mid S \cap Z \neq \emptyset\}$$
for $Z$ clopen in $\mathbb{X}$. Coalgebras for the Vietoris functor turn out to correspond exactly to descriptive general frames known from the literature on modal logic.  

More generally, it is possible to define for every set functor $T$ a ``Stone companion'' $\widehat{T} : \mathbf{Stone} \rightarrow \mathbf{Stone}$, in a way that generalizes the Vietoris construction. One such construction is presented in \cite{KKP05}, making use of predicate liftings for the functor $T$. This approach would work for our purposes here, but we find it more elegant to make use of the nabla modality, similar to the construction used in \cite{VVV13}. We proceed as follows: recall that we are given a distributive law $\nabla : T^\omega Q \rightarrow QT$, provided by the lax extension $L$ (or simply the Barr extension of $T$, if $T$ preserves weak pullbacks). Now, for every Stone space $\mathbb{X} = (X,\tau)$ there is the inclusion 
$$\iota_\mathbb{X} : V P\mathbb{X} \rightarrow QX$$
and these inclusions form a natural transformation $\iota : P \rightarrow QU$. So we get a natural transformation
$$ \nabla U \circ T^\omega \iota : TVP \rightarrow QTU $$
We shall abuse notation and from now on simply write this natural transformation as $\nabla : T^\omega VP \rightarrow QTU$. Now, given a Stone space $\mathbb{X}$, let $\langle Im(\nabla_\mathbb{X}) \rangle$ be the subalgebra of $QTX$ generated by the elements of the form $\nabla_\mathbb{X} \varphi$, for $\varphi \in T^\omega V P \mathbb{X}$. It is easy to see that this extends to a functor
$$ \ddot{T} : \mathbf{Stone} \rightarrow \mathbf{BA}^{op}$$
by letting, for $h : \mathbb{X} \rightarrow \mathbb{Y}$, the map $\ddot{T} : \langle Im(\nabla_\mathbb{Y}) \rangle \rightarrow \langle Im(\nabla_\mathbb{X}) \rangle$ simply be the restriction of $QTUh : QTU\mathbb{Y} \rightarrow QTU \mathbb{X}$ to the subalgebra $\langle Im(\nabla_\mathbb{X}) \rangle$. That this map actually goes into the algebra $Im(\nabla_\mathbb{X}) \rangle$ follows by naturality of $\nabla$. We now simply set
$$\widehat{T} = S \circ \ddot{T}$$   
Alternatively, we may describe the functor $\widehat{T}$ as simply the dual of a functor on $\mathbf{BA}$ introduced in the paper \cite{KKV12}.  There, the functor $\tilde{T} : \mathbf{BA} \rightarrow \mathbf{BA}$, parametric in a set functor $T$, is constructed as follows: first, let $\mathcal{L}_0$ be the left adjoint to the forgetful functor from the category of algebras of signature $(\neg,\wedge,\vee)$ to $\mathbf{Set}$, so that $\mathcal{L}_0 X$ is the free algebra of Boolean terms generated by $X$. The natural transformation $\sigma : \mathcal{L}_0  \circ Q  \rightarrow Q$ has its component at a set $X$ defined by letting $\sigma_X : \mathcal{L}_0 Q X \rightarrow QX$ be provided by the co-unit of the adjunction $\mathcal{L}_0 \dashv U$. In other words, we extend the identity map on $QX$ to the map $\sigma_X $ using freeness of the term algebra generated by $QX$. For any set $X$, an element $\varphi \in \mathcal{L}_0 T^\omega \mathcal{L}_0 QX$ is called a \textit{one-step formula}, and its \textit{one-step semantics} $\|\varphi\|^X_1$ in a set $X$ is given by the set
$$ \sigma_{TX} \circ \mathcal{L}_0\nabla_X \circ \mathcal{L}_0T^\omega \sigma_{X}(\varphi) \subseteq TX $$
Given a Boolean algebra $A$, we define the algebra $\tilde{T} A$ to be 
$$\{\|\varphi\|_1^{USA} \mid  \varphi \in \mathcal{L}_0 T^\omega \mathcal{L}_0(PSA)\}$$
The action on morphisms is defined by setting, for $h : A \rightarrow B$,
$$\tilde{T}h(\| \varphi \|_1^{USA}) = \| \mathcal{L}_0 T^\omega \mathcal{L}_0 PS h (\varphi)\|_1^{USB}$$
The actual definition of the functor $\tilde{T}$ in \cite{KKV12} is different from the one given here, but the fact that their definition provides a functor that is naturally isomorphic to  the present one is an easy corollary of the ``one-step soundness and completeness'' theorem proved in \cite{KKV12}, Theorem  7.5. It is easy to verify the following:
\begin{proposition}
The functor $\widehat{T}$ is naturally isomorphic to $S \circ \tilde{T} \circ P$.
\end{proposition}
Since $\mathbf{BA}$ is a variety it is co-complete (a survey of properties of categories of algebras can be found in \cite{adamekrosicky}), so the initial sequence for the functor $\tilde{T}$ is well defined.  It is proved in \cite{KKV12} that the functor $\tilde{T}$ is finitary, hence the initial sequence for $\tilde{T}$ stabilizes at $\omega$. Since $\widehat{T}$ is dual to $\tilde{T}$, we conclude:
\begin{lemma}
\label{terminalsequence}
The terminal sequence for $\widehat{T}$ stabilizes at $\omega$. 
\end{lemma}
So the final coalgebra for $\widehat{T}$ always exists, and is given by the $\omega$-th entry of the terminal sequence. 

\section{Bisimulations for $\widehat{T}$-coalgebras}

\subsection{Neighbourhood bisimulations}

We shall now provide a sound and complete notion of bisimulation for $\widehat{T}$-coalgebras. Bisimulations for the Vietoris functor have been studied in \cite{BeFoVe10}. Every  coalgebra $(\mathbb{X},\alpha)$ for $\mathcal{V}$ comes with an underlying powerset coalgebra, given by $(X,U\alpha)$. Such a coalgebra is just a Kripke frame, and a Kripke bisimulation between powerset coalgebras $(X,\alpha)$ and $(Y,\beta)$ is defined to be a $\overline{\psf}$-bisimulation, where $\overline{\psf}$ is the Barr extension of $\psf$. A \textit{Vietoris bisimulation} between two given Vietoris coalgebras $(\mathbb{X},\alpha)$ and $(\mathbb{Y},\beta)$ is then defined to be a Kripke bisimulation between the underlying $\psf$-coalgebras, which in addition is topologically closed as a subspace of the product $\mathbb{X} \times \mathbb{Y}$. The main technical result in \cite{BeFoVe10} is that the \textit{topological closure of a Kripke bisimulation is always a Vietoris bisimulation}. As we shall see, this result generalizes to arbitrary set functors. 

First, we introduce two auxiliary relation liftings for the functor $P$:
\begin{definition}
Given a pair of Stone spaces $\mathbb{X}$ and $\mathbb{Y}$, and a binary relation $R \subseteq X \times Y$, we define the relation $\overrightarrow{R} \subseteq VP(\mathbb{X}) \times VP(\mathbb{Y})$ by
$$ A \overrightarrow{R} B \text{ iff } R[A] \subseteq B$$
Conversely, we define $\overleftarrow{R}$ by
$$ A \overleftarrow{R} B \text{ iff } R^\dagger[B] \subseteq A$$
\end{definition}
Here, $R[A] = \{v \mid \exists u \in A: u R v\}$ and $R^\dagger[B] = \{v \mid \exists u \in B: v R u\}$. Note that the operation $\rightarrow$ is \textit{antitone}: if $R \subseteq S$ then  $\overrightarrow{S} \subseteq \overrightarrow{R}$. The same holds for $\leftarrow$.


We can now introduce our notion of bisimulations for $\widehat{T}$-coalgebras (where $L$ is a given lax extension for $T$):
\begin{definition}
Let $(\mathbb{X},\alpha)$ and $(\mathbb{Y}, \beta)$ be two $\widehat{T}$-coalgebras, and $R \subseteq X \times Y$ a binary relation. Then $R$ is said to be a \textit{neighbourhood bisimulation} if, whenever $u R v$, we have for all $\varphi \in T^\omega P\mathbb{X}$ and all $\psi \in T^\omega P\mathbb{X}$: 
\begin{enumerate}
\item
If $\varphi L(\overrightarrow{R}) \psi$ then $\nabla_\mathbb{X}\varphi \in \alpha(u)$ implies $\nabla_\mathbb{Y}\psi \in \beta(v)$
\item If $\varphi L(\overleftarrow{R}) \psi$ then $\nabla_\mathbb{Y}\psi \in \beta(v)$ implies $\nabla_\mathbb{X}\varphi \in \alpha(u)$
\end{enumerate}
\end{definition}
The exact sense in which neighbourhood bisimulations generalize Vietoris bisimulations will be explained in the next subsection. Before that, we shall prove that neighbourhood bisimulations are always sound and complete for behavioural equivalence. 

From now on, given a space $\mathbb{X}$, let $\in_\mathbb{X} \subseteq X \times P \mathbb{X}$ denote the membership relation between elements of $X$ and clopens, and let $\subseteq_\mathbb{X}$ denote the subsethood relation between clopens of $\mathbb{X}$.
The easy proof of the following lemma is left to the reader:

\begin{lemma}
\label{monotone}
Let $\varphi$ and $\psi$ be members of $TP\mathbb{X}$ such that $\varphi (L \subseteq_\mathbb{X}) \psi$. Then $\nabla_\mathbb{X}\varphi \subseteq \nabla_\mathbb{Y}\psi$.
\end{lemma}


\begin{theorem}
Two states $u,v$ in $\widehat{T}$-coalgebras $(\mathbb{X},\alpha)$ and $(\mathbb{Y},\beta)$ respectively are behaviourally equivalent iff they are related by some neighbourhood bisimulation.
\end{theorem}

\begin{proof}
Let $R$ be any neighbourhood bisimulation. Since $\mathbf{Stone}$ is dually equivalent to the algebraic variety $\mathbf{BA}$, it has all co-filtered limits, so the terminal sequence for $\widehat{T}$ can be constructed, and since $\widehat{T}$ preserves co-filtered limits it stabilizes at $\omega$.
We prove by induction that $R \subseteq \sim_\xi$ for each finite ordinal $\xi < \omega$, where $u \sim_\xi v$ iff $beh^\alpha_\xi (u) = beh^\beta_\xi(v)$. The case for $\xi = 0$ is trivial, so we treat the case for $\xi + 1$ given that the induction hypothesis holds for $\xi$.

 Now, suppose $u R v$. We want to show that $u \sim_{\xi + 1} v$, i.e. $beh^\alpha_{\xi + 1}(u) = beh^\beta_{\xi + 1}(v)$. Equivalently, we need to show that 
$$\widehat{T}(beh^\alpha_\xi)(\alpha(u)) = \widehat{T}(beh^\beta_\xi)(\beta(v))$$
For this, it suffices to prove that, for every $\theta \in T^\omega P(\mathcal{Z}_\xi)$, we have 
$$T^\omega P(beh^\alpha_\xi)(\theta) \in \alpha(u) \text{ iff } T^\omega P(beh^\beta_\xi)(\theta) \in \beta(v)$$
Since $R$ is a neighbourhood bisimulation, it suffices to show that, for every $\theta \in T^\omega P(\mathcal{Z}_\xi)$, we have
$$( T^\omega P(beh^\alpha_\xi)(\theta),T^\omega P(beh^\beta_\xi)(\theta) ) \in L\overrightarrow{R} \cap L\overleftarrow{R}$$
We only prove the statement for  $L\overrightarrow{R}$ since the proof of the second statement is symmetric.
We have  $R \subseteq \sim_\xi$ by the inductive hypothesis, hence $\overrightarrow{\sim_\xi} \subseteq \overrightarrow{R}$, and since the relation lifting $L$ is monotone, it suffices to show that
$$( T^\omega P(beh^\alpha_\xi)(\theta),T^\omega P(beh^\beta_\xi)(\theta) ) \in L(\overrightarrow{\sim_\xi})$$
for each $\theta$. Since $L$ is a symmetric lax extension for $T$, it is easy to show that
 $$(T^\omega P(beh^\alpha_\xi)(\theta),T^\omega P(beh^\beta_\xi)(\theta)) \in L(P(beh^\alpha_\xi);(P(beh^\beta_\xi))^\dagger)$$
where $;$ denotes relation composition, $\dagger$ gives the converse of a relation and we identify a mapping with its graph. So by monotonicity of $L$, it now suffices to prove that 
$$P(beh^\alpha_\xi);(P(beh^\beta_\xi))^\dagger \; \subseteq \; \overrightarrow{\sim_\xi}$$
So let $p$ be a clopen in $\mathbb{X}$ and $q$ a clopen in $\mathbb{Y}$ with 
$$( p,q ) \in P(beh^\alpha_\xi);(P(beh^\beta_\xi))^\dagger $$
This means that there is a clopen set $c$ in $\mathcal{Z}_\xi$ with $(beh^\alpha_\xi)^{-1}[c] = p$ and $(beh^\beta_\xi)^{-1}[c] = q$. Hence, we get

\begin{displaymath}
\begin{array}{lcl}
\sim_\xi[p] & = & (beh^\beta_\xi)^{-1}[(beh^\alpha_\xi)[p]] \\
& =  & (beh^\beta_\xi)^{-1}[(beh^\alpha_\xi)[(beh^\alpha_\xi)^{-1}[c]]]\\
& \subseteq &  (beh^\beta_\xi)^{-1}[c] \\
& = & q
\end{array} \end{displaymath}
This ends the proof of soundness for behavioural equivalence.

Conversely, let $f : (\mathbb{X},\alpha) \rightarrow (\mathbb{Z},\gamma)$ and $g : (\mathbb{Y},\beta) \rightarrow (\mathbb{Z},\gamma)$ be a co-span of $\widehat{T}$-coalgebra morphisms. We show that the pullback $R$ of $f$ and $g$ (in $\mathbf{Set}$) is a neighbourhood bisimulation.  So, suppose $u R v$, and let $\varphi \in T^\omega P\mathbb{X}$ and $\psi \in T^\omega P\mathbb{Y}$ be such that $\varphi (L \overrightarrow{R}) \psi$. Suppose that $\nabla_\mathbb{X}\varphi \in \alpha(u)$; we show that $\nabla_\mathbb{Y}\psi \in \beta(v)$.  First, note that we have $Tf[\nabla_\mathbb{X}\varphi] \in \gamma(f(u))$, because
$$Tf[\nabla_\mathbb{X}\varphi] \in \gamma(f(v)) \; \Leftrightarrow \; QTf(Tf[\nabla_\mathbb{X} \varphi]) \in \alpha(u) $$
and we have
$$\nabla_\mathbb{X} \varphi \subseteq QTf(Tf[\nabla_\mathbb{X} \varphi]) $$
and $\alpha(u)$ is a filter, hence closed under supersets. If follows that $\nabla_\mathbb{Z}Tf'(\varphi) \in \gamma(f(u))$, where $f' : P \mathbb{X} \rightarrow P\mathbb{Y}$ is the map defined by the assignment $p \mapsto f[p]$. To see this, since $\gamma(f(v))$ is a filter it suffices to prove that
$$Tf[\nabla_\mathbb{X}\varphi]  \subseteq \nabla_\mathbb{Z}Tf'(\varphi) $$
For this, we need to show that for all $a \in \nabla_\mathbb{X} \varphi$, we have $Tf(a) \in \nabla_\mathbb{Z}Tf'(\varphi)$. Equivalently, we need to show that $a (L\in_\mathbb{X}) \varphi$ implies $Tf(a)(L\in_\mathbb{Z}) Tf'\varphi$. Now, $a (L \in_\mathbb{X}) \varphi $ implies
$$ Tf(a) (Lf)^\dagger;(L\in_\mathbb{X});(Lf') Tf'(\varphi)$$
But we have
\begin{displaymath}
\begin{array}{lcl}
(Lf)^\dagger;(L\in_\mathbb{X});(Lf')  & = & (Lf^\dagger);(L\in_\mathbb{X});(Lf') \\
& = & L(f^\dagger;\in_\mathbb{X};f')\\
& \subseteq & L(\in_\mathbb{Z})
\end{array}
\end{displaymath}
Where the inclusion follows from the fact that
$$f^\dagger;\in_\mathbb{X};f' \; \subseteq\;  \in_\mathbb{Y}$$
So $\nabla_\mathbb{Z}Tf'(\varphi) \in \gamma(f(u))$, and it follows that 
$$QTg(\nabla_\mathbb{Z}Tf'(\varphi) \in \gamma(f(u))) = \nabla_\mathbb{Y}(T(Pg\circ f')\varphi) \in \beta(v) $$
We want to show that
$$ \nabla_\mathbb{Y}(T(Pg\circ f')\varphi)   \; (L\subseteq_\mathbb{Y}) \;\nabla_\mathbb{Y} \psi $$
so that we can apply Lemma \ref{monotone} to conclude that $\nabla_\mathbb{Y} \psi \in \beta(v)$. Since we know that $\varphi (L\overrightarrow{R}) \psi$, it suffices to prove that
$$(T(Pg \circ f'))^\dagger ; L\overrightarrow{R} \; \subseteq \; L\subseteq_\mathbb{Y}$$
Since $T(Pg \circ f'))^\dagger \subseteq L((Pg \circ f')^\dagger)$, it suffices in turn to show that
$$(Pg \circ f')^\dagger ; \overrightarrow{R} \; \subseteq \; \subseteq_\mathbb{Y}$$
But this is easily derived from the definition of $R$ as the pullback of $f$ and $g$.  

For the other direction, we can prove by a symmetric argument that if $\varphi \in T^\omega P\mathbb{X}$,  $\psi \in T^\omega P\mathbb{Y}$ and $\varphi (L\overleftarrow{R}) \psi$, then $\nabla_\mathbb{Y}\psi \in \beta(v)$ implies $\nabla_\mathbb{X} \varphi \in \alpha(u)$. 
\end{proof}

\subsection{Relating neighbourhood bisimulations to Vietoris Bisimulations}

We now show neighbourhood bisimulations for $\widehat{\psf}$ relate to Vietoris bisimulations. As one would hope,  Vietoris bisimulations can be recovered as closed neighbourhood bisimulations. For an exposition of the basic facts about nets in topology that we use in the proof, see \cite{bredon}.

Generally, given a $T$-coalgebra $(X,\alpha)$ and a Stone space $\mathbb{X} = (X,\tau)$, let $(\mathbb{X},\alpha_\mathbb{X})$ be a $\widehat{T}$-coalgebra with the property that
$$\alpha_\mathbb{X}(u) = \{\varphi \in \langle Im(\nabla_\mathbb{X}) \rangle \mid \alpha(u) \in \varphi\}$$
We call such a $\widehat{T}$-coalgebra a $\widehat{T}$-\textit{companion} to $(X,\alpha)$.
Using this terminology, the following proposition relates neighbourhood bisimulations to $L$-bisimulations:
\begin{proposition}
\label{sufficient}
Suppose $R$ is an $L$-bisimulation between $(X,\alpha)$ and $(Y,\beta)$, and suppose $(\mathbb{X},\alpha_\mathbb{X})$ and $(\mathbb{Y},\beta_\mathbb{Y})$ are $\widehat{T}$-companions to $(X,\alpha)$ and $(Y,\beta)$ respectively. Then $R$ is a neighbourhood bisimulation between $(\mathbb{X},\alpha_\mathbb{X})$ and $(\mathbb{Y},\beta_\mathbb{Y})$.
\end{proposition}

Given a set $X$ and a subset $Z \subseteq X$, we write $\Diamond Z$ as a short-hand for $\nabla_\mathbb{X}(\{Z,X\})$, which is a subset of $\psf X $. We have, for $\alpha \in \psf X$, $\alpha \in \Diamond Z$ iff $\alpha \cap Z \neq \emptyset$. Dually, define $\Box Z = \nabla_X (\{Z\}) \cup \nabla_X(\emptyset)$. (Essentially, these are just the usual box- and diamond modalities.) The following lemma is an exercise in basic topology:

\begin{lemma}
Let $\mathbb{X}$ and $\mathbb{Y}$ be Stone spaces, and let $R$ be a closed subset of the product space. Then for every closed set $Z$ in $\mathbb{X}$, the image $R[Z] $ is closed.
\end{lemma}


The relationship between neighbourhood bisimulations and Vietoris bisimulations is now established in the following result:
\begin{theorem}
\label{comparison}
Let $(\mathbb{X},\alpha)$ and $(\mathbb{Y},\beta)$ be a pair of Vietoris coalgebras, let $R$ be any closed relation in $\mathbb{X} \times \mathbb{Y}$ and let $(\mathbb{X},\alpha_\mathbb{X})$ and $(\mathbb{Y},\beta_\mathbb{Y})$ be $\widehat{\psf}$-companions to the coalgebras $(X,U\alpha)$ and $(Y,U\beta)$ respectively. Then  $R$ is a Vietoris bisimulation between $(\mathbb{X},\alpha)$ and $(\mathbb{Y},\beta)$ iff it is a neighbourhood bisimulation between $(\mathbb{X},\alpha_\mathbb{X})$ and $(\mathbb{Y},\beta_\mathbb{Y})$.
\end{theorem}

\subsection{Topological closure of neighbourhood bisimulations}

We now prove the main technical result of the paper:

\begin{theorem}
\label{main}
Let $B$ be any neighbourhood bisimulation between $(\mathbb{X},\alpha)$ and $(\mathbb{Y},\beta)$. Then the topological closure of $B$ in the product space $\mathbb{X} \times \mathbb{Y}$ is a neighbourhood bisimulation too.
\end{theorem}
\begin{proof}
Let $\overline{B}$ denote the topological closure of $B$.
Suppose that there is $x \in X$, $y \in Y$ and $\varphi \in T^\omega P\mathbb{X}$, $\psi \in T^\omega P\mathbb{Y}$ such that
\begin{itemize}
\item $x \overline{B} y$
\item $\varphi (L\overrightarrow{\overline{B}}) \psi$, but
\item $\nabla_\mathbb{X} \varphi \in \alpha(x)$ but $\nabla_\mathbb{Y} \psi \notin \beta(y)$ (or vice versa, but the other case is treated symmetrically.)
\end{itemize} 
We aim to derive a contradiction, using the fact that $B$ was assumed to be a neighbourhood bisimulation. First, since $(x,y) \in \overline{B}$, there is a net $\mathcal{N} = (x_i,y_i)_{i \in I}$ in $B$ for some directed poset $I$, such that $(x,y)$ is a limit point of this net. Define $\mathcal{N}_1$ to be the net $(x_i)_{i_I}$ and define $\mathcal{N}_2$ to be the net $(y_i)_{i \in I}$. Then $\mathcal{N}_1$ has $x$ as a limit point, and similarly $\mathcal{N}_2$ converges to $y$ (since the projection maps are continuous). Before we derive our contradiction, we prove the following claims:
\begin{description}
\item[Claim 1:] There exists an index $h \in I$ such that, for all $k \geq h$, we have
 $$\nabla_\mathbb{X} \varphi \in \alpha(x_k)$$
\item[Claim 2:] There exists an index $m \in I$ such that, for all $k \geq m$, we have
$$\nabla_\mathbb{Y} \psi \notin \beta(y_m)$$
\end{description}
We begin by proving Claim 1. Consider the net $(\alpha(x_i))_{i \in I}$ in $\widehat{T}\mathbb{X}$. Since the map $\alpha : \mathbb{X} \rightarrow \widehat{T}\mathbb{X}$ is a morphism in the category of Stone spaces, it is a continuous map, and since $x$ is a limit point of the net $(x_i)_{i \in I}$ it follows that $\alpha(x)$ is a limit point of the net $(\alpha(x_i))_{i \in I}$. This means that, for every  neighbourhood $S$ of $\alpha(x)$ in $\widehat {T}\mathbb{X}$, there is some index $h \in I$ such that $\alpha(x_k) \in S$ for all $k \geq h$. But by definition of the topology on the space $\widehat{T}\mathbb{X}$, if we  set
$$S = \{F \in \widehat {T}\mathbb{X} \mid \nabla_\mathbb{X} \varphi \in F\}$$
then $S$ is a clopen basis set, and certainly $\alpha(x) \in S$. Hence, we find an appropriate index $h$. 

For Claim 2, we consider the net $(\beta(y_i))_{i \in I}$. This net converges to $\beta(y)$, so for every neighbourhood $S$ of $\beta(y)$ in $\widehat{T}\mathbb{Y}$ there is some index $m \in I$ with $\beta(y_k) \in S$ for all $k \geq m$. With this in mind, set
$$S = \{F \in \widehat{T}\mathbb{Y}\mid \nabla_\mathbb{Y} \psi \notin F\}$$
This set is equal to
$$ \{F \in \widehat{T}\mathbb{Y} \mid (TY)\setminus \nabla_\mathbb{Y}\psi \in F\}$$
and so it is a clopen basis set in $\widehat {T}\mathbb{Y}$. Since $\beta(y) \in S$, we find an appropriate index $m$ to witness Claim 2.

With Claims 1-2 in place, we can finish the proof: let $h$ and $m$ be as described in these claims. Since $I$ is a directed set, there exists a common upper bound $k$ of $h$ and $m$ in $I$. Thus, we have $\nabla_\mathbb{X} \varphi \in \alpha(x_k)$, but $\nabla_\mathbb{Y}\psi \notin \beta(y_k)$. Since $(x_k,y_k) \in B$ and $B$ was a $\nabla$-bisimulation, to derive a contradiction at this point we only need to show that
$$\varphi (L \overrightarrow{B}) \psi$$
But we have assumed that  $\varphi (L \overrightarrow{\overline{B}}) \psi$, and since $B \subseteq \overline{B}$ and the operation $\rightarrow$ is antitone we get $\overrightarrow{\overline{B}} \subseteq \overrightarrow{B}$. Hence we get
$$ L\overrightarrow{\overline{B}} \; \subseteq \; L \overrightarrow{B}$$
Hence $\varphi (L\overrightarrow{B}) \psi$ as required, and the proof is done.
\end{proof}

As an immediate corollary to this result, we can derive the main technical result in \cite{BeFoVe10}:

\begin{corollary}[\cite{BeFoVe10}, Theorem 5.5]
The topological closure of a Kripke bisimulation between Vietoris coalgebras is a Vietoris bisimulation.
\end{corollary}

\begin{proof}
Let $R$ be any Kripke bisimulation between a pair of Vietoris coalgebras $(\mathbb{X},\alpha)$ and $(\mathbb{Y},\beta)$. It is easy to see that $(X,U\alpha)$ and $(Y,U\beta)$ have unique $\widehat{\psf}$-companions $(\mathbb{X},\alpha_\mathbb{X})$ and $(\mathbb{Y},\beta_\mathbb{Y})$ respectively, and $R$ is a neighbourhood bisimulation between $(\mathbb{X},\alpha_\mathbb{X})$ and $(\mathbb{Y},\beta_\mathbb{Y})$ by Proposition \ref{sufficient}. So the closure $\overline{R}$ of $R$ is also a neighbourhood bisimulation by Theorem \ref{main}. But since $\overline{R}$ is certainly closed, this means that it is a Vietoris bisimulation between $(\mathbb{X},\alpha)$ and $(\mathbb{Y},\beta)$, by Theorem \ref{comparison}.
\end{proof}

We can also derive a generalization of another result from \cite{BeFoVe10}, showing that the Vietoris bisimulations between two $\mathcal{V}$-coalgebras always form a complete lattice:
\begin{corollary}[\cite{BeFoVe10}, Corollary 5.7]
Let $(\mathbb{X},\alpha)$ and $(\mathbb{Y},\beta)$ be $\widehat{T}$-coalgebras. Then the family of topologically closed neighbourhood bisimulations between these coalgebras forms a complete lattice ordered by inclusion.
\end{corollary}

\begin{proof}
It is easy to check that the union of any family of neighbourhood bisimulations is a neighbourhood bisimulation. So let $S$ be any collection of closed neighbourhood bisimulations. The meet of these neighbourhood bisimulations is given by the topological closure of the union of all neighbourhood bisimulations contained in $\bigcap S$, and the join is given by the topological closure of $\bigcup S$.
\end{proof}
\section{Conclusion and Future Work}
In this paper we presented an alternative construction for lifting a set coalgebra to its Stone companion using the Moss-style $\nabla$-modality. We introduced a notion of bisimulation, namely neighbourhood bisimulation, for Stone coalgebras, and showed that it is sound and complete with respect to behavioural equivalence. Our notion of bisimulation generalizes Vietoris bisimulation introduced in \cite{BeFoVe10}--indeed, in the concrete case of the powerset functor, we showed that closed neighbourhood bisimulations are same as Vietoris bisimulations. We proved that  neighbourhood bisimulations are closed under topological closure, which is the main result of our paper. The main result in \cite{BeFoVe10}, stating that the topological closure of a Kripke bisimulation is a Vietoris bisimulation, follows as a corollary of our main result.

The Vietoris construction was originally presented in the context of compact Hausdorff spaces \cite{Vi22}. Therefore, a natural direction would be to generalize our results to coalgebras on compact Hausdorff spaces. We were able to show that neighbourhood bisimulations are sound and complete with respect to behavioural equivalence for coalgebras on compact Hausdorff spaces. However, the proof of closure of neighbourhood bisimulations under topological closure does not generalize for compact Hausdorff spaces. We leave this an interesting open problem.

\bibliographystyle{abbrv}
\bibliography{uepl_arxiv}

\section*{\appendixname}

\paragraph{Proof of Lemma \ref{monotone}}
\begin{proof}
Given that the premise holds, we need to show that $a (\overline{T}\in_\mathbb{X}) \varphi$ implies $a (\overline{T}\in_\mathbb{X}) \psi$. So suppose $a (\overline{T}\in_\mathbb{X}) \varphi$. Then 
$$a (\overline{T}\in_\mathbb{X}); (\overline{T}\subseteq_\mathbb{X}) \psi $$
Since clearly $\in_\mathbb{X};\subseteq_\mathbb{X} = \in_\mathbb{X}$, we get 
$$(\overline{T}\in_\mathbb{X} );(\overline{T}\subseteq_\mathbb{X}) = \overline{T}(\in_\mathbb{X};\subseteq_\mathbb{X}) = \overline{T}\in_\mathbb{X}$$
and so $a (\overline{T}\in_\mathbb{X}) \psi$ as required.
\end{proof}
\paragraph{Proof of Proposition \ref{sufficient}}
\begin{proof}
We prove only one direction of the condition required for $R$ to be a neighborhood bisimulation, since the other direction is proved in the same way. So suppose $u R v$, $\nabla_\mathbb{X}\varphi \in \alpha_\mathbb{X}(u)$ and $\varphi L \overrightarrow{R} \psi$. We want to show that $\psi \in \beta_\mathbb{Y}(v)$. Note that 
$$\nabla_\mathbb{X}\varphi \in \alpha_\mathbb{X}(u) \Leftrightarrow \alpha(u) \in \nabla_\mathbb{X} \varphi \Leftrightarrow \alpha(u) L(\in_\mathbb{X}) \varphi$$
and similarly for $\beta_\mathbb{Y}(v)$ and $\psi$. So we know that $ \alpha(u) L(\in_\mathbb{X}) \varphi$ and we want to show that $ \beta(v) L(\in_\mathbb{Y}) \psi$. We also know that $\alpha(u) LR \beta(v)$, since $R$ was assumed to be an $L$-bisimulation. We get
$$\beta(v)\; (LR^\dagger) \;\alpha(u)\;( L\in_\mathbb{X}) \; \varphi \;(L\overrightarrow{R}) \; \psi$$
and so  
$$\beta(v) L(R^\dagger;\in_\mathbb{X} ;\overrightarrow{R} ) \psi$$
Thus, it suffices to prove that
$$R^\dagger;\in_\mathbb{X} ;\overrightarrow{R} \; \subseteq \;\in_\mathbb{Y}$$
and the reader can easily check this to be true.
\end{proof}

\paragraph{Proof of Theorem \ref{comparison}}

\begin{proof}
Since $R$ is a Vietoris bisimulation it is a Kripke bisimulation between the underlying Kripke frames $(X,\alpha)$ and $(Y,\beta)$. So by Proposition \ref{sufficient} it is a neighbourhood bisimulation  between $(\mathbb{X},\alpha_\mathbb{X})$ and $(\mathbb{Y},\beta_\mathbb{Y})$.

For the converse direction, suppose $R$ is a  neighbourhood bisimulation, and let $u R v$.  We show that for every $u' \in \alpha(u)$, there is $v' \in \beta(v)$ with $u ' R v'$ (the converse direction is proved in the same manner).   We first prove the following:
\begin{description}
\item[Claim:] For any clopen $Z$ in $\mathbb{X}$, if $\alpha(u) \in \Diamond Z$ then $\beta(v) \in \Diamond R[Z]$ 
 \end{description}
Proof of the Claim: suppose $\alpha(u) \in \Diamond Z$. Since $Z$ is clopen, $R[Z]$ is closed, so 
$$R[Z] = \bigcap\{Z' \supseteq R[Z] \mid Z' \text{ clopen in } \mathbb{Y}\}$$ 
Let us write $F$ for the family of clopen sets $Z' \supseteq R[Z]$. We need to show that $\bigcap F \cap \beta(v) \neq \emptyset$. Suppose that  $\bigcap F \cap \beta(v) = \emptyset$; since these are all closed sets, there must be a finite family $F' \subseteq F$ such that $\bigcap F' \cap \beta(v) = \emptyset$. But then,  $\bigcap F'$ is a clopen set and $R[Z] \subseteq \bigcap F'$, hence  $\beta(v) \in \Diamond \bigcap F'$ since $R$ was a neighbourhood bisimulation. This is a contradiction with $\bigcap F' \cap \beta(v) = \emptyset$.

We now carry out the main proof: pick $w \in \alpha(u)$. Then for every clopen $Z$ with $w \in Z$, we have $\alpha(u) \in \Diamond Z$, and so  we have $\beta(v) \in \Diamond R[Z]$ by the previous Claim.  So there is some $s_Z \in Z$ and some $t_Z \in \beta(v) $ with $s_Z R t_Z$. We construct a net in $R$ as follows: let $(I,\supseteq)$ be the directed set with $I$ being the clopen neighbourhoods of $w$, and $\supseteq$ being the converse of the inclusion order. Then $(s_Z,t_Z)_{Z \in I}$ is a net in $R$. Since $R$ is a compact Hausdorff space, there exists a subnet of $(I,\supseteq)$ that converges to some limit point. That is, there is a directed set $(D,\leq)$ and a monotone map $F : (D,\leq) \rightarrow (I,\supseteq)$ such that the image of $F$ is cofinal in $(I,\supseteq)$, and such that the net
$(s_{F(a)},t_{F(a)})_{a \in D}$ converges to some limit point $(p,q) \in R$.

The net $(t_{F(a)})_{a \in D}$ then converges to $q$ in $\mathbb{Y}$ since the projection map $\pi_\mathbb{Y}$ is continuous, and so since this is a net in $\beta(v)$ and $\beta(v)$ is closed, we get $q \in \beta(v)$. It only remains to show that $p = w$. But since limits of nets in a Hausdorff space are unique and since the net $(s_{F(a)})_{a \in D}$ converges to $p$, it suffices to show that this net converges to $w$. So suppose $O$ is some open neighbourhood of $w$. Then, since $\mathbb{X}$ is a Stone space and so has a clopen basis, there is a clopen set $S \subseteq O$ with $w \in S$. By cofinality, there is some $a \in D$ with $F(a) \subseteq S$. We show that, for every $b \geq a$, we have $s_{F(b)} \in O$: by definition, $s_{F(b)} \in F(b)$, and since $$F(b) \subseteq F(a) \subseteq S \subseteq O$$
we get $s_{F(b)} \in O$ as required. 
\end{proof}

\end{document}